\documentclass{ebarticle}
\usepackage{ebutf8,apimac}

\title{Quantitative testing semantics for non-interleaving}
\author{Emmanuel Beffara}

\DeclareMathOperator\action{act}
\let\?=\mathcal

\begin{document}
\maketitle

\begin{abstract} 
  This paper presents a non-interleaving denotational semantics for the
  π-calculus.
  The basic idea is to define a notion of test where the outcome is not only
  whether a given process passes a given test, but also in how many different
  ways it can pass it.
  More abstractly, the set of possible outcomes for tests forms
  a semiring, and the set of process interpretations appears as a module over
  this semiring, in which basic syntactic constructs are affine operators.
  This notion of test leads to a trace semantics in which traces are
  partial orders, in the style of Mazurkiewicz traces, extended with readiness
  information.
  Our construction has standard may- and must-testing as special cases.
\end{abstract} 

\tableofcontents

\section{Introduction} 

The theory of concurrency has developed several very different models for
processes, focusing on different aspects of computation.
Process calculi are an appealing framework for describing and analyzing
concurrent systems, because the formal language approach is well suited to
modular reasoning, allowing to study sophisticated systems by means of
abstract programming primitives for which powerful theoretical tools can be
developed.
However, the vast majority of the semantic studies on process calculi like the
π-calculus have focused on the so-called interleaving operational semantics,
which is the basic definition of the dynamic of a process: the interaction of
a program with its environment is reduced to possible sequences of
transitions, thus considering that parallel composition of program components
is merely an abstraction that represents all possible ways of combining
several sequential processes into one.

There is clearly something unsatisfactory in this state of things.
Although sophisticated theories have been established for interleaving
semantics, most of which are based on various forms of bisimulation, they
fundamentally forget the crucial (and obvious) fact that concurrent processes
are intended to model situations where some events may occur independently.
Attempts at recovering this notion of independence in existing theories have
been made, for instance in the form of subtle variations on bisimulation or by
fully abstract encodings of non-interleaving semantics into interleaving ones
(in particular in Sangiorgi's work on locality and
causality~\cite{sg96:interleaving,bs98:causality}).
More recently, the old idea of Winskel's interpretation of CCS in event
structures~\cite{win82:ccs} has been revisited by Crafa, Varacca and Yoshida
to provide an actually non-interleaving operational semantics for the
π-calculus, using extensions of event structures~\cite{cvy07:event}.

This paper presents an attempt at defining a semantics for the π-calculus that
is both non-inter\-leaving (sometimes called “truly concurrent”) and
denotational, in the sense that the internal dynamics of a process is hidden,
and only the part that is observable by other processes is kept.
These two requirements may seem contradictory: “denotational” as we mean it
leads to the definition of testing semantics, which in turn leads to trace
semantics, which is very interleaving in nature.
Indeed,
consider the prototypical case of $a|b$ versus $a.b+b.a$:
how is it possible to distinguish them when looking at their interactions?
Both can do $a$ then $b$ or $b$ then $a$, but in the first case the paths
$a.b$ and $b.a$ are in fact one same run since $a$ and $b$ are independent,
while in the second case they correspond to two actually different choices.
We solve the contradiction by elaborating on this simple idea: instead of
checking whether a given process passes a given test, we check in how many
different ways it can pass it.
The word “different” here refers to different choices being made in situations
of non-determinism, and not simply different orderings of the same actions.

The approach presented here follows previous work by the
author~\cite{bef08:apc} on the search for algebraical\-ly pleasant denotational
semantics of process calculi.
The first step was to introduce in the π-calculus an additive structure (a
formal sum with zero) that represents pure non-determinism, and this technique
proved efficient enough to provide a readiness trace semantics (in the style
of Olderog and Hoare~\cite{oh86:spec}) with a complete axiomatization of
equivalence for finite terms.
The second step presented here further extends the space of processes with
arbitrary linear combinations, giving a meaning to these combinations in terms
of quantitative testing.
This introduction of scalar coefficients was not possible in the interleaving
case, because of the combinatorial explosion that arose even when simply
composing independent traces; quotienting by homotopy is the proper solution
to this problem.
Growing the space of processes to get more algebraic structure is motivated by
the idea that better structured semantics gives cleaner mathematical
foundations for the object of study, in the hope that the obtained theory will
be reusable for different purposes and that it will benefit from existing
mathematical tools.

\paragraph{Outline.}

In section~\ref{sec-operational}, we define the calculus on which our study is
built: a finite form of the πI-calculus.
An non-interleaving operational semantics is defined as follows: transitions
are those of the standard calculus, decorated with the position of each
action involved in a given transition, so that transitions are independent if
they derive from actions at independent positions.
Two execution paths are then considered homotopic if they differ only by
permutation of independent actions.
This technique is a variant of proved transitions introduced by Boudol and
Castellani~\cite{bc88:ccs} and notably studied by Degano and
Priami~\cite{dp92:proved,dp99:non-interleaving}.

In section~\ref{sec-testing}, the notion of test is defined.
Outcomes are taken from a semiring $\K$ in which multiplication represents the
parallel composition of independent results and addition represents the
combination of outcomes from different (non homotopic) runs.
Processes are equivalent if they yield the same outcome in all contexts.
The space of process equivalence classes appears a $\K$-module, on which the
outcome is a linear form, and syntactic constructs are affine operators.

In section~\ref{sec-algebraic}, we derive a first denotational semantics of
processes as linear forms over this space, in a construction similar to that
of the theory of distributions.
This construction provides an abstract interpretation of recursive processes
without having to include them in the initial construction of tests.

In section~\ref{sec-traces}, we further describe the space of finite processes
by showing that every finite process is a linear combination of traces.
Our notion of trace is an asynchronous variant of the traces induced by
standard semantics: they are partially ordered finite sets of actions,
augmented with readiness information.
This provides a second, more concrete denotational semantics that illustrates
the expressiveness of our notion of test.

Finally, in section~\ref{sec-classic}, we show that standard forms of test are
particular cases of our construction, obtained by choosing an appropriate
semiring for outcomes.

\paragraph{Future and related works.}

The present work is by no means a complete study of quantitative testing
semantics and its possible applications, but rather a presentation of the
basic ideas and their consequences.
A first objective is to clarify the relationships between the two proposed
semantics, possibly by establishing that traces form a basis of the space of
processes (maybe by using a ring or field instead of a semiring for outcomes).
Another challenging direction for future work is using the linear-algebraic
interpretation for specification of processes, using tools like differential
equations to specify behaviours; this should provide a reconstruction of the
semantics on arbitrary vector spaces instead of the concrete space of
processes, which could be a way to a new family of denotational semantics for
process calculi.

Along with this long-term ideas, it is naturally interesting to extend our
work with more features in the calculus.
A more precise account of recursion is a desirable thing: surely infinite
behaviours fit in our framework, but the present work does not study it in
full detail for lack of space.
External choice is a natural feature to add in the framework, but previous
work~\cite{bef08:apc} suggests that it is painless.
Unrestricted name passing, on the other hand, is a more delicate matter, and
we believe that getting a satisfactory understanding of the more regular case
of internal mobility first is necessary to handle it.

Several works by other authors are related to the present work.
Crafa, Varacca and Yoshida's event structure semantics probably has very
strong relationships with our trace semantics: it has to be expected that
their event structures can be used as an intermediate between the process
calculus and the traces, and that traces and outcomes can be deduced from
configurations of the event structures.
The operational semantics and its similarity to Mazurkiewicz traces also
suggests that relations could be made with more abstract semantics, like
Melliès and Mimram's asynchronous
games~\cite{mel04:innocence,mm08:concurrent}.
Previous work on the search for algebraic semantics of processes include
Boreale and Gadducci's processes as formal series~\cite{bg06:series}, which
has notable similarities with the present work, although their work is carried
out in CSP.
Finally, strong relationships are expected with differential interaction
nets~\cite{er04:diffnets,el07:pidiff}, which have linear algebraic semantics
and are expressive enough to encode the π-calculus.

\section{Parallel operational semantics} 
\label{sec-operational}

We consider the π-calculus with internal mobility, or πI-calculus, extended
with a parallel composition without interaction and with \emph{outcomes} from
a commutative semiring~$\K$.
We consider the monadic variant of the calculus for simplicity, but using the
polyadic form would not pose any significant problem.
The most important point is that we consider finite processes, without
recursion, for the construction of our framework, and we handle potentially
infinite behaviours in a second phase in section~\ref{sec-algebraic}.
\begin{definition}\label{def-terms}
  We assume a countable set $\Names$ of names.
  Polarities are elements of $\Pola=\implem{\shpos,\shneg}$.
  Terms are generated by the following grammar:
  \begin{syntax}
  \define[actions] α
    \case u^ε(x) \comment{with $u,x∈\Names$ and $ε∈\Pola$}
  \define[processes] P, Q
    \case k  \comment{outcome, with $k∈\K$}
    \case α.P  \comment{action}
    \case \place P  \comment{performed action}
    \case P|Q  \comment{parallel composition with interaction}
    \case P\para Q  \comment{parallel composition without interaction}
    \case\new{x}P  \comment{hiding}
  \end{syntax}
  Terms are considered up to injective renaming of bound names and commutation
  of restrictions, i.e. $\new{x}\new{y}P=\new{y}\new{x}P$, with the standard
  convention that all bound names are distinct from all other names.
\end{definition}

The parallel composition without interaction allows us to write a term like
$a\para\bar{a}$ which can perform the dual actions $a$ and $\bar{a}$
independently but does not allow them to synchronize — this slightly extends
the expressiveness of the calculus but not in a dramatic way, and it
simplifies the theory.

The prefix $\place$ represents an action that already occurred.
It has no computational meaning but has the effect that the positions of
actions in the terms are preserved when reducing, which will simplify
definitions below.

We want to define an operational semantics in which commutation of independent
transitions is allowed.
In order to make this possible by only looking at transition labels, we have
to enrich the labels so that different occurrences of a given action are
distinguishable.
We do this by simply introducing in each label the positions in the syntax
tree of all actions involved (as a consequence, the operational semantics
cannot be defined up to structural congruence).
\begin{definition}
  A position is a finite sequence of integers.
  The concatenation of $ι$ and $κ$ is written $ι.κ$, the empty position is
  written $ε$.
  The prefix order is written $≤$ and two positions $ι$ and $κ$ are
  independent (written $ι\indep κ$) if they are incomparable.
\end{definition}

\begin{definition}
  Transition labels can be of one of two kinds:
  \begin{syntax}
  \define a, b
    \case u^ε(x):ι \comment{visible action}
    \case (ι,κ)  \comment{internal transition}
  \end{syntax}
  For a label $a$ and a position $ι$, $ι.a$ denotes the label $a$ where each
  position $κ$ is replaced with $ι.κ$.
  Transitions are derived by the rules of table \ref{table-transitions}.

  An interaction is finite sequence of transition labels.
  A path is a finite sequence of internal transition labels.
  An interaction $p=a_1a_2…a_n$ is valid for $P$, written $p∈P$, if there are
  valid transitions $P\trans{a_1}P_1\trans{a_2}\cdots\trans{a_n}P_n$.
\end{definition}
\begin{table}
  \centering
    \RULE0{ α.P \trans{α:ε} \place P }
    \DISP
  \hfil
    \AX{ P \trans{a} P' }
    \RULE1{ \place P \trans{1.a} \place P' }
    \DISP
  \hfil
    \AX{ P \trans{u^ε(x):ι} P' }
    \AX{ Q \trans{u^{¬ε}(y):κ} Q' }
    \RULE2{ P|Q \trans{(1.ι,2.κ)} \new{x}(P'|Q'[x/y]) }
    \DISP
  \hfil
    \AX{ P \trans{a} P' }
    \AX{ x ∉ a }
    \RULE2{ \new{x}P \trans{a} \new{x}P' }
    \DISP
  \hfil
    \AX{ P \trans{a} P' }
    \RULE1{ P|Q \trans{1.a} P'|Q }
    \DISP
  \hfil
    \AX{ P \trans{a} P' }
    \RULE1{ Q|P \trans{2.a} Q|P' }
    \DISP
  \hfil
    \AX{ P \trans{a} P' }
    \RULE1{ P \para Q \trans{1.a} P' \para Q }
    \DISP
  \hfil
    \AX{ P \trans{a} P' }
    \RULE1{ Q \para P \trans{2.a} Q \para P' }
    \DISP
  \smallskip
  \caption{Transition rules}
  \label{table-transitions}
\end{table}

This technique can be seen as a version of Boudol and Castellani's proved
transitions~\cite{bc88:ccs,dp99:non-interleaving} simplified for our purpose.
It is clear that for all term $P$ and interaction $p∈P$, there is exactly one
term $P/p$ such that there is a transition sequence $P\trans{p}P/p$ (up to
renaming of revealed bound names).
Remark that by removing all positions from labels (replacing $(ι,κ)$ by $τ$)
one gets the standard labeled transition system for the πI-calculus.

\begin{definition}\label{def-homotopy}
  Two labels $a$ and $b$ are independent (written $a\indep b$) if all
  positions in $a$ are independent of all positions in $b$.
  Homotopy is the smallest congruence $\homo$ over paths such that
  $ab\homo ba$ when $a\indep b$.
\end{definition}

Two execution paths of a given term are homotopic if it is possible to
transform one into the other by exchanging consecutive transitions if they are
independent.
Prefixing generates local constraints which propagate to paths by this
relation.
A first remark is that transition labels contain enough information so that
homotopy does not depend on the term in which paths are taken.

\begin{proposition}\label{prop-homotopy}
  For all term $P$ and all interactions $p,q$ such that $p\homo q$, $p∈P$ if
  and only if $q∈P$, and then $P/p=P/q$.
\end{proposition}
\begin{proof}
  The basic case is $p=ab$ and $q=ba$ for some $a$ and $b$ with $a\indep b$.
  We thus prove that for any pair of transitions $P\trans{a}Q\trans{b}R$ with
  $a\indep b$, there is a term $Q'$ such that $P\trans{b}Q'\trans{a}R$.
  Proceed by induction on the derivation of $P\trans{a}Q$.
  \begin{itemize}
  \item 
    The case of the action rule $α.P\trans{α:ε}\place P$ cannot happen since no
    position is independent of $ε$ but $b$ is supposed to be independent of
    $α:ε$.
  \item 
    In the case of the $\place$ rule, we have $\place P\trans{1.a}\place Q\trans{1.b}\place R$ and
    we can proceed by induction on $P\trans{a}Q\trans{b}R$.
  \item 
    In the case of the interaction rule, we have $a=(1.ι,2.κ)$ for some
    positions $ι$ and $κ$, and the second transition starts from
    $\new{x}(P'|Q'[x/y])$.
    Reason by case analysis on the shape of this second transition.
    If all positions in $b$ start with $1$, then $b=1.b'$ for some $b'$ and
    the second transition comes from $P'\trans{b'}P''$, so we can apply the
    induction hypothesis on $P\trans{u^ε(x):1.ι}P'\trans{b'}P''$ to get
    transitions $P\trans{b'}R\trans{u^ε(x):1.ι}P''$, from which we deduce
    $P|Q\trans{b}R|Q\trans{a}\new{x}(P''|Q')$.
    If all positions in $b$ start with $2$, the same argument applies, the
    substitution of $x$ for $y$ is innocuous since it only affects actions
    prefixed by $u^{¬ε}(y)$.
    If $b=(1.ι',2.κ')$ for some $ι'$ and $κ'$, then we have
    $P'\trans{v^η(x')ι'}P''$ and $Q'\trans{v^{¬η}(y')κ'}Q''$, so we
    can apply the induction hypothesis on $P$ and $Q$ independently, from
    which we deduce
    $P|Q\trans{b}\new{x'}(P'''|Q'''[x'/y'])
    \trans{a}\new{x'}\new{x}(P''|Q''[x'/y',x/y])$,
    which concludes this case since the substitutions $[x/y]$ and $[x'/y']$
    are independent and the order of restrictions is irrelevant.
  \item 
    The $\new{x}$ context rule is obvious.
  \item 
    In the right context rule for parallel composition, we have $a=1.a'$ for
    some $a'$.
    If $b=2.b'$ for some $b'$, then $a$ occurs in $P$ and $b$ occurs in $Q$,
    so they obviously commute.
    If $b$ has the form $α:1.ι$, then we proceed by induction in $P$.
    If $b$ is a label $(ι,κ)$ with one of $ι,κ$ starting with $1$, then we
    proceed by induction on the visible action at this position, in a
    similar way as for parallel composition.
  \item 
    The other context rules for composition with and without
    interaction are similar.
  \end{itemize}
  The general case follows.
\end{proof}

\begin{definition}
  A pre-trace is a homotopy class of interactions.
  A run is a homotopy class of maximal paths.
  The sets of pre-traces and runs of a term $P$ are written $\Pretraces(P)$
  and $\Runs(P)$ respectively.
  The unique reduct of a term $P$ by a pre-trace $ρ$ is written $P/ρ$.
\end{definition}

Runs are the intended operational semantics: they are complete
executions of a given system, forgetting unimportant interleaving of actions
and remembering only actual ordering constraints.
A pre-trace can be seen as a Mazurkiewicz trace on the infinite language of
transition labels, with the independence relation from
definition~\ref{def-homotopy}, except that, because of our transition rules
(and because of the use of the place-holder $\place$), each label occurs at
most once in any interaction.
A crucial fact is that pre-traces are uniquely defined by the set of
their labels:
\begin{proposition}\label{prop-permutation}
  Let $p$ and $q$ be two interactions of a term $P$ such that $p$ and $q$ are
  permutations of each other, then $p\homo q$.
\end{proposition}
\begin{proof}
  We first prove that for all interaction $a_1…a_nb∈P$ such that $b∈P$ we have
  $a_1…a_nb\homo ba_1…a_n$, by induction on $n$.
  The case $n=0$ is trivial.
  For the case $n≥1$, remark that the hypothesis implies $a_1\indep b$: if
  some position in $a_1$ was less than a position in $b$ then $b$ could only
  occur after $a_1$, which contradicts $b∈P$, and $a_1∈P$ also implies that no
  position in $b$ is less than a position in $a$.
  Therefore we have $ba_1∈P$ and $ba_1\homo a_1b$.
  Applying the induction hypothesis on $P/a_1$ yields $ba_2…a_n\homo a_2…a_nb$
  from which we conclude.
  The case of arbitrary permutations follows by recurrence on the length of
  $p$ and $q$.
\end{proof}

\begin{definition}\label{def-causal-order}
  Let $P$ be a term and $ρ∈\Pretraces(P)$.
  By proposition~\ref{prop-permutation}, $ρ$ is identified with the set of its
  labels.
  The causal order in $ρ$ is the partial order $≤_ρ$ on labels in $ρ$ such
  that $a≤_ρb$ if $a=b$ or $a$ occurs before $b$ in all interactions in~$ρ$.
\end{definition}

This presentation is much simpler to handle than
explicit sets of runs, so this is the one we will mainly use.
Interactions that constitute a given pre-trace are simply the topological
orderings of this partially ordered set of transitions.
Traces are a further quotient of pre-traces, defined and
studied in section~\ref{sec-traces}.

\section{Quantitative testing} 
\label{sec-testing}

We now define a form of observation based on interaction, in the style of
testing equivalences, that takes homotopy into account.
Standard testing naturally leads to interleaving semantics, so we have to
refine our notion of test, and that is what outcomes are for.
The set $\K$ is a semiring in order to
represent two ways of combining results: the product is the parallel
composition of independent results and the sum is the combination of
results from distinct runs.

\begin{definition}
  The state $s(P)∈\K$ of a term $P$ is the product of all outcomes in active
  position in $P$:
  \begin{center}
    $ s(k) := k $, \hfil
    $ s(α.P) := 1 $, \hfil
    $ s(\place P) := s(\new{x}P) := s(P) $, \hfil
    $ s(P|Q) := s(P\para Q) := s(P) \, s(Q) $.
  \end{center}
  The outcome of a term $P$ is $\outcome{P}=\sum_{ρ∈\Runs(P)}s(P/ρ)$.
  Two terms $P$ and $Q$ are observationally equivalent, written $P≃Q$, if
  $\outcome{P|R}=\outcome{Q|R}$ for all~$R$.
%
\end{definition}

Classic forms of test intuitively correspond to the case where $\K$ is the set
of booleans for the two outcomes success and failure, with operations defined
appropriately.
This particular case is detailed in section~\ref{sec-classic}.

\subsection{Basic properties} 

\begin{theorem}\label{thm-congruence}
  Observational equivalence is a congruence.
\end{theorem}
\begin{proof}
  Consider a pair of equivalent processes $P≃Q$.
  Let $α$ be an arbitrary action, we first prove that $α.P≃α.Q$.
  Let $R$ be an arbitrary process.
  The set $\Runs(α.P|R)$ can be split into two parts:
  the set $\?R_0$ of runs where the action $α$ is not triggered and
  the set $\?R_1$ of runs in which it is.
  Then for each run $ρ∈\?R_1$, there is a position $ι$ such that $(1,2.ι)∈ρ$.
  Let $ρ_1$ be the partial run $\set{a}{a∈ρ,a≤_ρ(1,2.ι)}$, that is the minimal
  run that triggers $α$; we have $(α.P|R)/ρ_1=\new{x}(P|R')$ for some $R'$;
  let $ρ_2=ρ∖ρ_1$, so that $ρ_2$ is a run of $P|R'$ and
  $(α.P|R)/ρ=\new{x}(P|R')/ρ_2$.
  Let $\?S$ be the set of triples $(ρ_1,R',ρ_2)$ for all $ρ∈\?R_1$.
  Obviously $\Runs(α.P|R)$ is in bijection with $\?R_0\uplus\?S$ and
  \[
    \outcome{α.P|R}
    = \sum_{2.ρ∈\?R_0} s(R/ρ)
    + \sum_{(ρ_1,R',ρ_2)∈\?S} s((P|R')/ρ_2)
  \]
  Now let $\?L=\set{(ρ_1,R')}{∃ρ_2,(ρ_1,R',ρ_2)∈\?S}$, and let $(ρ_1,R')∈\?L$.
  Since $\?R_1$ contains all runs of $α.P|R$ that trigger $α$, it contains all
  the runs of $P|R'$ since $P|R'$ can be reached from $α.P|R$, so we have
  $\set{ρ_2}{(ρ_1,R',ρ_2)∈\?S}=\Runs(P|R')$, hence
  \[
    \sum_{(ρ_1,R',ρ_2)∈\?S} s((P|R')/ρ_2)
    = \sum_{(ρ_1,R')∈\?L} \sum_{ρ_2∈\Runs(P|R')} s((P|R')/ρ)
    = \sum_{(ρ_1,R')∈\?L} \outcome{P|R'}
  \]
  By hypothesis, for all $R'$ we have $\outcome{P|R'}=\outcome{Q|R'}$ so
  \[
    \outcome{α.P|R}
    = \sum_{rρ∈\?R_0} s(R/ρ)
    + \sum_{(ρ_1,R')∈\?L} \outcome{Q|R'}
    = \outcome{α.Q|R}
  \]
  since the reasoning above equally applies to $Q$.
  Therefore we get $α.P≃α.Q$.

  \smallbreak

  For parallel composition, let $R$ and $S$ be arbitrary terms, we want to
  prove $\outcome{(P|R)|S}=\outcome{(Q|R)|S}$, in order to get $P|R≃Q|R$.
  Let $φ$ be the function over positions such that for all $ι$,
  $φ(1.1.ι)=1.ι$, $φ(1.2.ι)=2.1.ι$ and $φ(2.ι)=2.ι$,
  and for all path $p$, let $φ(p)$ be the path obtained by applying $φ$ on all
  positions in $p$.
  Then $φ$ is a bijection between the paths of $(P|R)|S$ and those of
  $P|(R|S)$, and it preserves homotopy so it actually provides a bijection
  between $\Runs((P|R)|S)$ and $\Runs(P|(R|S))$.
  Moreover, for all $ρ∈\Runs((P|R)|S)$, we have
  $s(((P|R)|S)/ρ)=s((P|(R|S))/φ(ρ))$, so $\outcome{(P|R)|S}=\outcome{P|(R|S)}$.
  Similarly we get $s((Q|R)|S)=\outcome{Q|(R|S)}$, and by hypothesis we have $P≃Q$ so
  $\outcome{P|(R|S)}=\outcome{Q|(R|S)}$, from which we conclude.

  \smallbreak

  For parallel composition without interaction, let $R$ and $S$ be arbitrary
  terms, we want to prove $\outcome{(P\para R)|S}=\outcome{(Q\para R)|S}$, in order to get
  $P\para R≃Q\para R$.
  The technique used for parallel composition with interaction does not apply
  here, because there is no simple form of associativity between the two
  parallel compositions.
  However, if the free names of $P$ and $R$ are disjoint, it is easily seen
  that $P|R$ and $P\para R$ are equivalent, and this is the fact we will use
  here.

  Let $A$ be a set of pairs $(ι,κ)$ where $ι$ is the position of an action in
  $R$ and $κ$ is the position of an action in $S$, such that these actions are
  on a free name and may interact with each other, and such that $A$ is a
  partial injection (each position of $R$ occurs at most once on the left,
  each position of $S$ occurs at most once on the right).
  Call this kind of set a synchronization and let $\?S$ be the set of all
  synchronizations.
  We say that a run $ρ∈\Runs((P\para R)|S)$ satisfies $A$, written $ρ⊩A$, if
  the interactions between $R$ and $S$ in $ρ$ are exactly those designated by
  $A$, that is if $\set{(ι,κ)}{(1.2.ι,2.κ)∈ρ}=A$.

  We will define $R^A$ and $S^A$ to be rewritings of $R$ and $S$ such that the
  pairs $(ι,κ)$ are guaranteed to interact.
  For this purpose, for each $(ι,κ)∈A$, let $a_{ι,κ}$ and $w_{ι,κ}$ be fresh
  names: $a_{ι,κ}$ is a new name on which the pair will interact (in order to
  avoid conflicts with other names) and $w_{ι,κ}$ will act as a witness of
  $(ι,κ)$, that will ensure that the pair actually interacts.
  Define $R^A$ as the term $R$ in which each action $a^ε(x).T$ at a
  position $ι$ such that there is an $(ι,κ)∈A$ is replaced by
  $a_{ι,κ}^{ε}(x).(T|w_{ι,κ}.1)$, and define $S^A$ as the term $S$ in
  which each action $a^ε(x).T$ at a position $κ$ such that there is an
  $(ι,κ)∈A$ is replaced by $a_{ι,κ}^{ε}(x).T$ (without $w_{ι,κ}$).
  Let $W_A$ be any parallel composition of $w_{ι,κ}.0|\bar{w}_{ι,κ}.1$ for all
  ${ι,κ}∈A$.

  If we now examine the runs of $(P|R^A)|(S^A|W_A)$, we observe that if a run
  $ρ$ does not trigger the actions of a given pair $(ι,κ)∈A$, then it must
  contain the reduction of $w_{1,j}.0|\bar{w}_{ι,κ}.1$ into $\place0|\place1$
  (because runs are made of maximal paths), so the outcome of this run is $0$.
  On the other hand, if all the interactions given by $A$ occur in $ρ$, it is
  still possible that some of the $w_{ι,κ}.0|\bar{w}_{ι,κ}.1$ reduce into
  $\place0|\place1$, but there is one possibility that each $\bar{w}_{ι,κ}.1$
  interacts with the $w_{ι,κ}.1$ in $R^A$.
  From a run that satisfies this condition, we can deduce a unique run of
  $(P\para R)|S$ that satisfies $A$, and reciprocally from a run of
  $(P\para R)|S$ that satisfies $A$ we can deduce a unique run of
  $(P|R^A)|(S^A|W_A)$ that does not reduce any $w_{ι,κ}.0$.
  Moreover, this bijection between runs preserves outcomes, so
  $\outcome{(P|R^A)|(S^A|W_A)}$ is the sum of all $s(((P\para R)|S)/ρ)$ for
  all runs $ρ$ that satisfy $A$.
  From this we get the following decomposition:
  \begin{multline*}
    \outcome{(P\para R)|S}
    = \sum_{A∈\?S} \sum_{ρ⊩A} s(((P\para R)|S)/ρ)
    = \sum_{A∈\?S} \outcome{(P|R^A)|(S^A|W_A)} \\
    = \sum_{A∈\?S} \outcome{P|(R^A|(S^A|W_A))}
    = \sum_{A∈\?S} \outcome{Q|(R^A|(S^A|W_A))}
    = \outcome{(Q\para R)|S}
  \end{multline*}
  The equality $\outcome{(P|R^A)|(S^A|W_A)}=\outcome{P|(R^A|(S^A|W_A))}$ is
  justified by the same argument as above for parallel composition, and the
  substitution of $Q$ for $P$ is the hypothesis $P≃Q$.
  The final equality is the same reasoning for $Q$ as for $P$ above.

  The equality $\outcome{\new{x}P|R}=\outcome{\new{x}Q|R}$ is justified by the
  fact that $\outcome{\new{x}P|R}$ and $\outcome{P|R}$ are equal if the name
  $x$ is fresh with respect to $R$.
\end{proof}

\begin{table}
  \begin{tabular}{llll}
    commutativity &
      $ P | Q  ≃  Q | P $ &
      $ P \para Q  ≃  Q \para P $ \\
    associativity &
      $ (P | Q) | R  ≃  P | (Q | R) $ &
      \rlap{$ (P \para Q) \para R  ≃  P \para (Q \para R) $} \\
    neutrality &
      $ P | 1  ≃  P $ &
      $ P \para 1  ≃  P $ \\[1ex]
    scope commutation &
      $ \new{x}\new{y} P  ≃  \new{y}\new{x} P $ \\
    scope extrusion &
      $ \new{x} (P | Q)  ≃  P | \new{x} Q $
      & with $x∉\fn(P)$ \\
    scope neutrality &
      $ \new{x} k  ≃  k $ \\[1ex]
    non-interaction &
      $ (P \para Q) | R  ≃  (P | R) \para Q $
      & with $\fn(Q)\cap\fn(R)=∅$ \\[1ex]
    place-holder &
      $ \place P ≃ P $ \\[1ex]
    inaction &
      $ \new{u} u^ε(x).P ≃ 1 $ \\
    non-interference &
      $ \new{u}( u(x).P | \bar{u}(x).Q ) ≃ \new{ux}(P|Q) $
  \end{tabular}
  \smallskip
  \caption{Basic equivalences.}
  \label{table-basic}
\end{table}
\begin{proposition}\label{prop-basic}
  The equivalences of table~\ref{table-basic} hold.
\end{proposition}
\begin{proof}
  For commutativity of composition with interaction, consider three terms
  $P,Q,R$. We first establish a bijection between $\Runs((P|Q)|R)$ and
  $\Runs((Q|P)|R)$.
  Let $φ$ be the function over positions that exchanges $ll$ and $1.2$ at the
  beginning of words, \ie\ $φ(1.1.ι)=1.2.ι$, $φ(1.2.ι)=1.1.ι$ and $φ(2.ι)=2.ι$ for all
  position $ι$.
  For all path $p$, define $φ(p)$ as the path obtained by applying $φ$ to all
  positions in $p$.
  Then clearly, for all path $p$ of $(P|Q)|R$, $φ(p)$ is a path of $(Q|P)|R$.
  Moreover, for all paths $p$ and $q$, $p\homo q$ if and only if
  $φ(p)\homoφ(q)$.
  Therefore $φ$ is a bijection between $\Runs((P|Q)|R)$ and $\Runs((Q|P)|R)$.
  Besides, since $\K$ is commutative, we have
  $\outcome{(P|Q)|R/p}=\outcome{(Q|P)|R/φ(p)}$
  for all $P,Q,R$ and $p$, so we have $\outcome{(P|Q)|R}=\outcome{(Q|P)|R}$.

  For associativity, we use the same technique with relabeling function
  defined as $φ(1.1.1.ι)=1.1.ι$, $φ(1.1.2.ι)=1.2.1.ι$, $φ(1.2.ι)=1.2.2.ι$ for
  all $ι$ and $φ(κ)=κ$ for all other positions $κ$; the final argument is
  associativity of the product in $\K$.
  For neutrality, we use $φ(1.1.ι)=φ(1.ι)$ for all $ι$ and $φ(κ)=κ$ for all
  other positions $κ$, and conclude by the neutrality of $1$ in $\K$.

  For parallel composition without interaction, the same arguments
  apply.
  For the non-interaction rule, a similar argument applies, using the fact
  that there can never be any interaction between $Q$ and $R$ if they do not
  share any free name (this property is specific of the πI calculus).

  For the scoping rules, we simply remark that $\outcome{\new{x}P|Q}=\outcome{P|Q}$ if $x$
  is a fresh name, since names have no influence on outcomes.

  For the place-holder rule, $\outcome{\place P|Q}=\outcome{P|Q}$ is proved by applying the
  function $φ$ such that $φ(1.1.ι)=1.ι$ for all $ι$ and $φ(κ)$ for all other
  $κ$.
  This establishes a bijection between $\Runs(\place P|Q)$ and $\Runs(P|Q)$
  since the $\place$ is never involved in any transition.

  For the inaction rule, remark that in a term $\new{u}u^ε(x).P|R$, there
  can be no transition involving $u^ε(x)$, hence all runs are made of
  labels of the form $(2.ι,2.κ)$, so the runs of $\new{u}u^ε(x).P|R$ are
  the runs of $R$ with an extra $2$ in front of each position, moreover the
  outcomes are the same since $s(\new{u}u^ε(x).P)=1$.

  For the non-interference rule, remark that all runs of
  $\new{u}(u(x).P|\bar{u}(x).Q)|R$ contain the transition $(1.1,1.2)$, because
  of maximality and the fact that $R$ cannot provide actions on $u$.
  The reduct by this transition is $\new{ux}(\place P|\place Q)|R$, and
  its runs are those of the original term without $(1.1,1.2)$, so it has the
  same outcome.
  We thus have
  $\new{u}(u(x).P|\bar{u}(x).Q)≃\new{ux}(\place P|\place Q)$,
  and the equivalence $\new{ux}(\place P|\place Q)≃\new{ux}(P|Q)$
  follows from the previous rules.
\end{proof}

The non-interaction rule is formulated as it is for generality.
Note that it implies the intuitive fact that the two compositions coincide for
terms with disjoint free names: if $\fn(P)\cap\fn(Q)=∅$ then
\[
  P|Q≃(1\para P)|Q≃(1|Q)\para P≃Q\para P≃P\para Q.
\]

Thanks to these properties, when considering processes up to observational
equivalence, we can consider the compositions to be associative and
commutative.
In this case we use the notation $\prod_{i∈I}P_i$ to denote the parallel
composition without interaction of the $P_i$ in any order (assuming only that
$I$ is finite).

\subsection{The space of processes} 

In order to study processes up to observational equivalence, we will now
describe some of the structure of the space of equivalence classes.
The first ingredient is to identify an additive structure that represents pure
non-determinism.

\begin{table}
  Commutative monoid for $⊕,0$:
  \begin{align*}
    P⊕Q &≃ Q⊕P &
    (P⊕Q)⊕R &≃ P⊕(Q⊕R) &
    P⊕0 &≃ P
  \end{align*}
  Action of the semi-ring $\K$:
  \begin{align*}
    1⋅P &≃ P &
    k_1k_2⋅P &≃ k_1⋅k_2⋅P \\\
    0⋅P &≃ 0 &
    (k_1+k_2)⋅P &≃ k_1⋅P ⊕ k_2⋅P &
    k⋅(P⊕Q) &≃ k⋅P ⊕ k⋅Q
  \end{align*}
  Bilinearity of compositions, linearity of hiding:
  \begin{align*}
    P|(Q⊕R) &≃ (P|Q)⊕(P|R) &
    P|(k⋅Q) &≃ k⋅(P|Q) \\
    P\para (Q⊕R) &≃ (P\para Q)⊕(P\para R) &
    P\para (k⋅Q) &≃ k⋅(P\para Q) \\
    \new{x}(P⊕Q) &≃ \new{x}P ⊕ \new{x}Q &
    \new{x}(k⋅P) &≃ k⋅\new{x}P
  \end{align*}
  \caption{$\K$-module laws over processes.}
  \label{table-module}
\end{table}

\begin{proposition}\label{prop-module}
  Let $\Procs_\K$ be the set of equivalence classes of processes over the
  semiring of outcomes $\K$.
  For all terms $P$ and $Q$ and all outcome $k$, define
  \begin{syntax}
    \define P⊕Q \case \new{u}((u.P|u.Q)|\bar{u}.1)
      \comment{where $u$ is a fresh name,}
    \define k⋅P \case k|P
  \end{syntax}
  Then $(\Procs_\K,⊕,0,⋅)$ is a $\K$-module,
  parallel compositions are bilinear operators and hiding is linear,
  i.e. the equivalences of table~\ref{table-module} hold.
\end{proposition}
\begin{proof}
  We first show that, for all terms $P$, $Q$ and $R$,
  $\outcome{(P⊕Q)|R}=\outcome{P|R}+\outcome{Q|R}$.
  Consider $\Runs((P⊕Q)|R)=\Runs(\new{u}((u.P|u.Q)|\bar{u}.1)|R)$.
  It is clear that any run contains an interaction of $\bar{u}.1$ with either
  $u.P$ or $u.Q$, since none of these may interact with anything else.
  We can thus write $\Runs((P⊕Q)|R)=\?R_1\uplus\?R_2$ where $\?R_1$ is the set
  of runs that contain $(1.1.1,1.2)$ and $\?R_2$ is the set of runs that contain
  $(1.1.2,1.2)$.
  The runs in $\?R_1$ are the runs of $\new{u}((u.P|1)|\bar{u})|R$ and each of
  these runs has the same outcome in both terms, so
  \[
    \sum_{ρ∈\?R_1} s\bigl(((P⊕Q)|R)/ρ\bigr)
    = \outcome{\new{u}((u.P|1)|\bar{u})|R}
    = \outcome{P|R}
  \]
  by the equivalences of table~\ref{table-basic}.
  By a similar argument, we get 
  the same for $\?R_2$ and $\outcome{Q|R}$
  so we finally get $\outcome{(P⊕Q)|R}=\outcome{P|R}+\outcome{Q|R}$.

  This equality and the fact that $(\K,+,0)$ is a commutative monoid
  implies that $(\Procs_\K,⊕,0)$ is a commutative monoid (where
  $0$ is the atomic term with outcome $0$).

  For any terms $P$ and $Q$ and any outcome $k$, it is clear that
  $\outcome{(k|P)|Q}=k\outcome{P|Q}$, since the term $k$ has no transition and
  contributes $k$ multiplicatively to all outcomes of the term.
  This directly implies that the operation $k⋅P$ has all required properties.

  For the bilinearity of compositions, consider arbitrary terms $P,Q,R,S$.
  By previous results, we have
  \begin{multline*}
    \outcome{(P|(Q⊕R))|S}
    = \outcome{(Q⊕R)|(P|S)}
    = \outcome{Q|(P|S)} + \outcome{R|(P|S)} \\
    = \outcome{(P|Q)|S} + \outcome{(P|R)|S}
    = \outcome{((P|Q)⊕(P|R))|S}
  \end{multline*}
  This proves that parallel composition distributes over $⊕$, and the fact
  that $0$ is absorbing is equivalent to the rule $0⋅P≃0$.
  The same rules for parallel composition without interaction can be proved by
  similar arguments about the partition of $\Runs(P\para(Q⊕R))$ into runs that
  choose $Q$ and runs that choose~$R$.

  For hiding, consider arbitrary terms $P,Q,R$ and let $x$ be a name.
  Assume (without loss of generality) that $x$ does not occur in $R$.
  Then we have
  \[
    \outcome{\new{x}(P⊕Q)|S}
    = \outcome{(P⊕Q)|S}
    = \outcome{P|S} + \outcome{Q|S}
    = \outcome{\new{x}P|S} + \outcome{\new{x}Q|S}
    = \outcome{(\new{x}P⊕\new{x}Q)|S}
  \]
  The equivalence $\new{x}0≃0$ is one of the rules of table~\ref{table-basic}.
\end{proof}

Remark that all syntactic constructions induce linear constructions
on equivalence classes, except for the action prefix, which is not linear but
actually affine.
Indeed, for an action $α$, the term $α.0$ is not equivalent to $0$: it will be
neutral in executions that do not trigger $α$, and multiply the outcome by
$0$ (thus annihilating it) in runs that do.
It can be understood as a statement “I could have performed $α$ but I will not
do it” so that any run that contradicts this statement has outcome $0$.
The purely linear part of actions is the opposite: the linear action
$\lin{α}.P$ will act as $α.P$ if its environment actually triggers the action,
but will turn to $0$ if it is never activated.

\begin{definition}\label{def-linear-action}
  For all action $α$ and term $P$, the linear action of $α$ on $P$ is
  \begin{syntax}
    \define \lin{α}.P
    \case \new{w}(α.(P|w.1)|(w.0|\bar{w}.1))
    \comment{where $w$ is a fresh name.}
  \end{syntax}
  An interaction is said to trigger the linear action if it triggers the
  action $w.1$.
  Terms of the form $α.0$ are called an inactions.
\end{definition}

This definition has the expected behaviour because of the maximality of runs.
If $\lin{α}.P$ is in active position, then any run that does not trigger $α$
must instead trigger $w.0$, hence any such run has outcome $0$.
A run in which the term $\lin{α}.P$ does not produce $0$ must activate $α$, so
that $w.1$ acts instead of $w.0$.

\begin{table}
  Linearity:
  \begin{align*}
    \lin{α}.(P⊕Q) &≃ \lin{α}.P⊕\lin{α}.Q &
    \lin{α}.(k⋅P) &≃ k⋅\lin{α}.P &
    \new{u}\lin{u}^ε(x).P &≃ 0
  \end{align*}
  Composition of inactions
  (the subject of $β$ is not bound by $α$):
  \begin{align*}
    \lin{α}.(β.0 | P) &≃ β.0 | \lin{α}.P &
    α.0|α.0 &≃ α.0 &
    α.0|\bar{α}.0 &≃ 0 \\
    \lin{α}.(β.0 \para P) &≃ β.0 \para \lin{α}.P &
    α.0\paraα.0 &≃ α.0
  \end{align*}
  \caption{Laws of linear actions and inactions.}
  \label{table-linear-actions}
\end{table}

\begin{proposition}\label{prop-affine-action}
  For all $α$ and $P$, $α.P≃\lin{α}.P⊕α.0$.
  The function $P\mapsto\lin{α}.P$ is linear and the equivalences of
  table~\ref{table-linear-actions} hold.
\end{proposition}
\begin{proof}
  We first prove $α.P≃\lin{α}.P⊕α.0$.
  Consider an arbitrary term $Q$ call $\?R_1$ the set of runs of $α.P|Q$ that
  contain a transition $(1,ι)$, i.e. runs that trigger $α$, and let $\?R_0$ be
  the set of runs that do not.
  The runs of $\?R_0$ are also runs of $α.0|Q$, moreover for each run
  $ρ∈\Runs(α.0|Q)∖\?R_0$ the action $α$ is triggered so $0$ contributes to the
  outcome and $s((α.0|Q)/ρ)=0$, hence
  $\sum_{ρ∈\?R_0}s((α.P|Q)/ρ)=\outcome{α.0|Q}$. Now consider a run $ρ∈\?R_1$.
  By definition, there is a $κ$ such that $(1,κ)∈ρ$.
  We deduce from $ρ$ a run $ρ'$ of $\lin{α}.P$ as follows:
  let $φ$ be the function that maps each $1.1.ι$ to $1.1.1.1.ι$ and all other
  positions to themselves; set $ρ':=φ(ρ∖(1,κ))∪\implem{(1.1,κ),(1.1.1.2,1.2.2)}$.
  This literally means that $ρ'$ is $ρ$ where all positions in $P$ are
  shifted to reflect their positions in
  $\new{w}(α.(P|w.1)|(w.0|\bar{w}.0))|Q$, $(1,κ)$ is shifted to reflect the
  new position of $α$, and the interaction between $w.1$ and $\bar{w}.1$ is
  added (which is valid since it is freed when $α$ is triggered).
  Clearly $ρ'$ is a run of $\lin{α}.P|Q$ and
  $s((\lin{α}.P|Q)/ρ')=s((α.P|Q)/ρ)$.
  The mapping $ρ\mapstoρ'$ is objective, and its image is the set of runs of
  $\lin{α}.P|Q$ that trigger $w.1$.
  By maximality, any other run of $\lin{α}.P|Q$ must trigger $w.0$, hence the
  outcome of all other runs is $0$, which implies
  $\outcome{\lin{α}.P|Q}=\sum_{ρ∈\?R_1}s((α.P|Q)/ρ)$.
  We can finally deduce $\outcome{α.P|Q}=\outcome{α.0|Q}+\outcome{\lin{α}.P|Q}$ and conclude.

  For linearity, we use the fact that $\outcome{\lin{α}.P|Q}$ is the sum of the
  $s((\lin{α}.P|Q)/ρ)$ for the runs $ρ$ that actually trigger $α$ (and the
  witness action $w.1$).
  If $P=k|P'$, these runs are the same in $\lin{α}.(k|P')|Q$ and
  $\lin{α}.(1|P')|Q$, but the outcomes are multiplied by $k$ in the first
  case, so $\outcome{\lin{α}.(k|P')|Q}=k⋅\outcome{\lin{α}.(1|P')|Q}$ and
  $\lin{α}.(k|P')≃k|\lin{α}.(1|P')≃k|\lin{α}.P'$.
  If $P=P_1⊕P_2$, the choice is eventually active in all relevant runs, so
  each of these runs triggers either $P_1$ or $P_2$.
  We can thus establish a bijection between $\Runs(\lin{α}.(P_1⊕P_2)|Q)$ and
  the disjoint union of $\Runs(\lin{α}.P_1|Q)$ and $\Runs(\lin{α}.P_2|Q)$.
  Since outcomes are preserved by this bijection, we finally get
  $\outcome{\lin{α}.(P_1⊕P_2)|Q}=\outcome{\lin{α}.P_1|Q}+\outcome{\lin{α}.P_2|Q}$
  and $(P_1⊕P_2)|Q≃(P_1|Q)⊕(P_2|Q)$.

  The equivalence $\new{u}u^ε(x).P≃0$ can be deduced from previous equations:
  \begin{align*}
    \new{u}u^ε(x).P
    &= \new{uw}(u^ε(x).(P|w.1) | (w.0 | \bar{w}.1)) \\
    &≃ \new{w}(\new{u}u^ε(x).(P|w.1) | (w.0 | \bar{w}.1)) \\
    &≃ \new{w}(1 | (w.0 | \bar{w}.1))
    ≃ \new{w}(w.0 | \bar{w}.1)
    ≃ \new{w}(0 | 1)
    ≃ 0
  \end{align*}

  For the equivalence $\lin{α}.(β.0|P)≃β.0|\lin{α}.P$,
  assuming the subject of $β$ is not the bound name of action $α$,
  let $Q$ be an arbitrary term and consider $\Runs(\lin{α}.(β.0|P)|Q)$.
  Any run that does not trigger $\lin{α}$ or that triggers both $\lin{α}$ and
  $β$ has outcome $0$, so the only relevant runs are those that trigger
  $\lin{α}$ but not $β$.
  Clearly these runs are in bijection with the runs of $(β.0|\lin{α}.P)|Q$
  that trigger $\lin{α}$ and not $β$, by a simple rewriting of the positions.
  Moreover, this bijection preserves outcomes, so the sums of the outcomes of
  these runs are the same.
  A similar argument proves $\lin{α}.(β.0\para P)≃β.0\para\lin{α}.P$.

  For the composition of inactions, the relevant runs of a term $(α.0|α.0)|P$
  or $(α.0\para α.0)|P$ are those that do not trigger any occurrence of $α$,
  so the number of such occurrences does not matter.
  Finally, we get $α.O|\bar{α}.0≃0$ by the remark that all runs of
  $(α.0|\bar{α}.0)|P$ must trigger one of the inactions: either $α.0$
  interacts with $P$, or $\bar{α}.0$ interacts with $P$, or none of these
  happen and $α.0$ and $\bar{α}.0$ must interact together, by maximality of
  runs.
\end{proof}

\begin{definition}\label{def-simple-term}
  A term is \emph{simple} if it is generated by the grammar
  \[
    P,Q := 1 ,\; α.0 ,\; \lin{α}.P ,\; (P|Q) ,\; (P\para Q) ,\; \new{x}P
  \]
  An pre-trace $ρ∈\Pretraces(P)$ is exhaustive if it triggers all linear
  actions and no inaction, and no sub-term of $P/ρ$ has the form $Q|R$ with
  $Q$ containing some $α.0$ and $R$ containing $\bar{α}.0$.
  The set of such pre-traces is written $\Pretraces_e(P)$.
\end{definition}


Simple terms have the property that the outcome of any run is either $1$ or
$0$.
More precisely, it is easy to see that the outcome of a run is $1$ if and only
if it triggers all linear actions and no inaction.
The notion of exhaustive pre-trace is the correct extension of this notion to
pre-traces, indeed every run of a simple term $P|Q$ with outcome $1$ is made
of an exhaustive pre-trace of $P$ and an exhaustive pre-trace of $Q$.
The condition on $P/ρ$ simply rules out interactions of $P$ that
lead to a term $P'$ where there are dual inactions that may interact, since
that would imply $P'≃0$, as a generalization of the equation
$α.0|\bar{α}.0≃0$.

Remark that, by the decomposition of proposition~\ref{prop-affine-action} and
the linearity of all constructions of simple terms, we immediately prove that
every term is equivalent to a linear combination of simple terms.
As a consequence, two terms $P$ and $Q$ are equivalent if and only if for all
\emph{simple} term $R$, $\outcome{P|R}=\outcome{Q|R}$.

\section{A linear algebraic semantics} 
\label{sec-algebraic}

The equivalence of finite processes is defined by the fact that they give the
same outcome when tested against the same finite processes.
The equivalence class of a term $P$ is thus completely defined by the function
$Q\mapsto\outcome{P|Q}$, which can be considered as a function from
equivalence classes to outcomes.
Moreover, by the properties of the space of processes, we know that this
function is linear.

\begin{definition}
  A behaviour is a linear form over $\Procs_\K$.
  A partial behaviour is a linear form defined over a submodule of $\Procs_\K$.
  The behaviour of a term $P$ is the form $⟦P⟧$ such that, for all
  $Q∈\Procs_\K$, $⟦P⟧(Q)=\outcome{P|Q}$.
\end{definition}

Switching from a space to its dual makes the space of considered objects grow,
as we will see below.
In our context, it allows us to move from inductive objects (finite processes)
to coinductive objects (intuitively, this includes infinite terms).
This technique is in some sense analogous to the basic idea of the theory of
distributions: consider a generalized function as a linear form over simple
well-behaved objects (smooth test functions, as analogous of our finite
terms).

We now describe a way of giving semantics to infinitary processes, showing
that recursive process definitions have solutions as partial behaviours.

\begin{definition}
  We assume we have a set $\Indets$ of process indeterminates of the form
  $X⟨x_1…x_n⟩$, which represent an unknown term with free names $x_1…x_n$.
  The set of partial terms is generated by the same grammar as finite
  processes (as of definition~\ref{def-terms}), augmented with indeterminates.
  The set of indeterminates of a partial term $P$ is written $\ind(P)$.

  The refinement preorder is the relation $⊑$ over partial terms such that
  $P⊑Q$ if $Q$ is obtained from $P$ by substituting each indeterminate by an
  arbitrary term with the same free names.
  The relation $⊑_f$ is its restriction to the case when the right-hand side
  is a finite term.
\end{definition}

\begin{definition}
  Let $P$ be a partial term.
  If there is a $k∈\K$ such that $\outcome{Q}=k$ for all finite $Q$ with
  $P⊑_fQ$, then we set $\outcome{P}=k$, otherwise $\outcome{P}$ is undefined.
  The interpretation $⟦P⟧$ is the partial function $Q\mapsto\outcome{P|Q}$
  from $\Procs_\K$ to $\K$.
  Two partial terms are equivalent if the have the same interpretation.
\end{definition}

This is clearly an extension of the semantics of total terms, since the set of
refinements of a total term $P$ is $\implem{P}$.
One easily checks that this definition of the interpretation of partial
terms enjoys the same properties as finite terms:
\begin{proposition}
  For all partial term $P$, $⟦P⟧$ is a partial behaviour.
  The equations of tables \ref{table-basic} and \ref{table-module} hold for
  partial terms.
  Interpretations are preserved by injective renaming of indeterminates.
  Equivalence is preserved by prefixing, hiding and
  composition with partial terms with distinct indeterminates.
\end{proposition}
\begin{proof}
  First, remark that for all $P⊑_fP'$, $\outcome{P'|0}=0$, hence
  $\outcome{P|0}$ so $⟦P⟧$ is defined on $0$.
  Now assume $⟦P⟧$ is defined on $Q$ and $R$, for all $P⊑_fP'$ we have
  $\outcome{P'|(Q⊕R)}=\outcome{P'|Q}+\outcome{P'|R}=⟦P⟧(Q)+⟦P⟧(R)$ so $⟦P⟧$ is
  correctly defined on $Q⊕R$.
  The same argument applies for $⟦P⟧(k⋅Q)=k⋅⟦P⟧(Q)$.
  As a consequence, $⟦P⟧$ is indeed a partial behaviour.

  Let $P,Q,R$ be partial terms.
  The relation
  \[
    \set{((P'|Q')|R',(Q'|P')|R')}{P⊑_fP',Q⊑_fQ',R⊑_fR'}
  \]
  is obviously a bijection between total refinements of $(P|Q)|R$ and total
  refinements of $(Q|P)|R$, and this bijection preserves outcomes because of
  the equation $P|Q≃Q|P$ for total terms.
  As a consequence we have $\outcome{(P|Q)|R}=\outcome{(Q|P)|R}$ for all $R$,
  which implies $⟦P|Q⟧=⟦Q|P⟧$.

  The same argument applies for all other equations.
  For scope extrusion and non-interaction, we use the fact that indeterminates
  have a fixed set of free names.
  For the equations in which a sub-term is duplicated (distribution of
  compositions over $⊕$), we use the fact that all occurrences of a given
  indeterminate are replaced by the same term when refining.

  If $P$ is a partial term and $Q$ is a refinement of $P$ obtained by
  injectively renaming the indeterminates of $P$, then for all total $R$ we
  clearly have $P⊑R$ if and only if $Q⊑R$, so $P≃Q$.

  Let $P,P',Q$ be partial terms such that $P≃P'$ and
  $\ind(P)\cap\ind(Q)=\ind(P')\cap\ind(Q)=∅$.
  Then the refinements of $P|Q$ are the compositions of a refinement of $P$
  and a refinement of $Q$, chosen independently since $\ind(P)\cap\ind(Q)=∅$. 
  Let $R$ be an arbitrary total term.
  If $\outcome{(P|Q)|R}$ is defined and has value $k$, then for all refinements
  $P⊑_fP''$ and $Q⊑_fQ''$ we have $\outcome{(P''|Q'')|R}=k$ so
  $\outcome{(P|Q'')|R}=k$,
  therefore $\outcome{(P'|Q'')|R}$ is defined and has value $k$, hence
  $\outcome{(P'|Q)|R}=k$ and $P'|Q≃P|Q$.
  Analogous reasoning yields $P'\para Q≃P\para Q$.
  Preservation of equivalence by other syntactic constructs is immediate.
\end{proof}

\begin{definition}
  An ideal is a non-empty set $A$ of partial terms that is downwards closed
  and such that for all $P,Q∈A$ there is an $R∈A$ such that $P⊑R$ and $Q⊑R$.
  The interpretation of an ideal $A$ is the upper limit of the
  interpretations of its elements, that is the partial function $⟦A⟧$
  such that for all $P$, $⟦A⟧(P)=k$ if $⟦Q⟧(P)=k$ for some $Q∈A$.
\end{definition}

The definition of $⟦A⟧$ is valid since interpretation of partial terms is
clearly increasing, when ordering partial terms by refinements and partial
functions by extension (or graph inclusion).
Note that for all total term $P$, the set $\set{Q}{Q⊑_fP}$ is an ideal that
has the same interpretation as $P$.

The set of ideals, ordered by inclusion, is not well founded:
if $(x_n)_{n∈\Nat}$ is an infinite family of names, then each set
$A_n=\set{X⟨x_1,…,x_i⟩}{i≥n}$ is an ideal and $\set{A_n}{n∈\Nat}$ is an
infinite descending chain.
However, if we restrict to a finite number of public names (which does not
change expressiveness, since bound names are not restricted), then the set
is well-founded, and the smallest ideal is the set $Ω$ of all
indeterminates with the set of all free names.

Every syntactic construction for total terms naturally induces a construction
for ideals, for instance the parallel composition $A|B$ is the downwards
closure is $\set{P|Q}{P∈A,Q∈B}$ for $P$ and $Q$ chosen with disjoint
indeterminates.
All these constructions are increasing for inclusion.
Moreover, the union of a directed set of ideals is an ideal, so every
equation $X=A(X)$ has a solution in ideals (the least fixed point of
$X\mapsto A(X)$, i.e. the union of the $A^n(Ω)$).
Consequently, all processes definable by recursion are interpreted  by partial
behaviours.

\section{Asynchronous traces} 
\label{sec-traces}

Simple terms remove one source of ambiguity in the meaning of processes: the
fact that each action may or may not be activated.
By linearity, they also reduce the computation of outcomes to the computation
of the number of non-zero outcomes.
However, they do not form a basis of observable process behaviours, because
they may contain internal transitions, which are not observable and can be a
source of non-determinism.

A trace as defined below can be seen as a deterministic simple term, up to
observational equivalence.
It has visible actions, with a partial order imposed by some internal
prefixing structure, and these actions may not interact with each other, only
with the environment; it also contains inactions, representing the fact that
the choice was made not to do some of the actions.

\begin{definition}\label{def-trace}
  A trace $T$ is a tuple $(\web{T},p,s,≤,N)$ where
  \begin{itemize}
  \item $\web{T}$ is a finite set (the events, or action occurrences),
  \item $p$ is a function from $\web{T}$ to $\Pola$ (the polarity),
  \item $s$ is a function from $\web{T}$ to $\Names\uplus\web{T}$
    (the subject),
  \item $\leq$ is a partial order over $\web{T}$ such that
    $∀a,b∈\web{T}$, if $s(b)=a$ then $a<b$,
  \item $N$ is a finite subset of $\Pola×(\Names\uplus\web{T})$
    (the inactions).
  \end{itemize}
\end{definition}

For an action $a∈\web{T}$, $s(a)$ is the subject of action $a$, that is
the channel on which $a$ happens: if $s(a)∈\Names$ then it is a public
channel, otherwise it is the private channel bound by the action
$s(a)$. The set $N$ indicates which actions could have been performed (in
parallel) after the trace has been consumed.

\subsection{Traces as basic terms} 


\begin{definition}
  Let $P$ be a simple term and let $ρ$ be an exhaustive pre-trace of $P$.
  The trace induced by $ρ$ is the trace $ρ^*$ such that
  \begin{itemize}
  \item $\web{ρ^*}$ is the set of visible transition labels in $ρ$,
  \item $p$ maps labels to their polarity,
  \item $s$ maps labels to their subject, either the name for public channels
    or the action that creates the name for private channels,
  \item $≤$ is the causal order (as of definition~\ref{def-causal-order})
    restricted to visible transition labels,
  \item $N$ is the set of all $(ε,u)$ such that $u^ε.0$ occurs in active
    position in $P/ρ$.
  \end{itemize}
\end{definition}

Note that the condition that $s(b)=a$ implies $a<b$ is satisfied by $ρ^*$,
because in our language the action prefixes are synchronous: in an action
$u(x).P$, the action $u(x)$ that binds $x$ is automatically a prefix of all
actions on $x$.
However, synchrony is not necessary for this property to hold:, the fact that
the name is bound is the important point: even if internal transitions can
occur on a bound name, visible transition are possible only after the name has
been revealed by the action it is bound to.

The definition above identifies the trace that is the observable content of a
pre-trace.
With some coding, we can prove that any trace can be implemented in the
calculus, in the sense that for every trace $T$ there is a term $\implem{T}$
that has a unique exhaustive pre-trace $ρ$, the content of which is $T$.

\begin{definition}
  Let $T$ be a trace.
  For all $a$ and $b$ in $\web{T}$, let $x_{ab},y_{ab},z_a$ be fresh names.
  For all $a∈\web{T}$, let $\action(a)=s(a)^{p(a)}(z_a)$ if $s(a)∈\Names$ and
  $\action(a)=z_{s(a)}^{p(a)}(z_a)$ if $s(a)∈\web{T}$ and define
  \[\textstyle
    A_a^T := \{\lin{x}_{ba}\}_{b<a} . \lin{\action}(a) .
      \biggl( \prod_{a<c} \linb{y}_{ac}
      \biggm\| \prod_{s(c)=a} A_c^T
      \biggm\| \prod_{(ε,a)∈N} z_a^ε.0
      \biggr)
  \]
  where $\{\lin{x}_{ba}\}_{b<a}$ represents a sequence of prefixes that
  contains all actions $\lin{x}_{ba}$ for all $b<a$, in any order.
  The implementation of $T$ is the process
  \[\textstyle
    \implem{ T } := \new{x_{ab}y_{ab}}_{a,b∈\web{T}}
      \biggl( \prod_{s(a)∈\Names} A_a^T
      \biggm| \prod_{a<b} \lin{y}_{ab}.\linb{x}_{ab}
      \biggr)
      \bigm\| \prod_{(ε,u)∈N ,\, u∈\Names} u^ε.0
  \]
\end{definition}

The intuition is the following:
each action in $T$ is translated by the linear action it describes, which
provides the right set of visible actions.
Inactions are translated straightforwardly.
The ordering is imposed by communication on internal names: for each
action $a$, the translation $\lin{\action}(a)$ is prefixed by a blocking
input $x_{ba}$ for each action $b<a$.
Activating this action frees the signals $x_{ac}$ for all $c>a$, which
guarantees that the order is respected.
We cannot implement this system one set of names $x_{ab}$, because the
actions $\lin{\action}(a)$ must be composed without
interaction, in order to avoid internal transitions between actions that are
supposed to implement visible transitions.
We thus split each signal into two names, $x_{ab}$ and $y_{ab}$, and put in
parallel (with interaction) a set of forwarders $y_{ab}.\bar{x}_{ab}$ that
performs the synchronization between signals.
If the subject of an action $a$ is the bound name of an action $b$, then
$\action(a)$ is put in the continuation of action $\action(b)$, which imposes
an order between this action; this is compatible with the constraint $b<a$
from the definition of traces.

\begin{proposition}\label{prop-trace}
  For all trace $T$, the term $\implem{T}$ is simple, has a unique
  exhaustive pre-trace $ρ$ and $ρ^*=T$.
\end{proposition}
\begin{proof}
  The fact that $\implem{T}$ is simple is obvious by definition of $\implem{T}$.

  Let us first build an exhaustive pre-trace of $\implem{T}$.
  Let $(a_i)_{1≤i≤n}$ be a topological ordering of $\web{T}$.
  We deduce a sequence of terms $(P_i)_{1≤i≤n+1}$ such that for each $i<n$
  there is an interaction from $P_i$ to $P_{i+1}$ made of a transition
  $\action(a_i):ι$ and internal transitions.
  Let $P_1=\implem{T}$.
  Let $i$ be an integer such that $1≤i≤n$, assume $P_i$ is a reduct of
  $\implem{T}$ that contains the $\lin{\action}(a_j)$ for all $j≥i$ and in
  active position all the $\linb{x}_{a_ja_k}$ such that $j<i≤k$ and $a_j<a_k$.
  Then the term $A_{a_i}^T$ occurs in active position in $P_i$ and the
  prefix $\{\lin{x}_{a_ja_i}\}_{a_j<a_i}$ can be consumed, which puts
  $\lin{\action}(a_i)$ in active position.
  We can then apply a transition $\action(a_i):ι$ for some $ι$ followed by an
  internal transition that consumes the $w.1$ contained in the linear action
  (as of definition~\ref{def-linear-action}).
  This puts in active position the $\linb{y}_{a_{i+1}a_j}$ for all
  $a_j>a_i$, and each of these can interact with the
  $\lin{y}_{a_ia_j}.\linb{x}_{a_ia_j}$, which puts in active position the
  $\linb{x}_{a_ia_j}$.
  By this interaction we reach a state $P_{i+1}$ that satisfies the condition
  we assumed on $P_i$.
  Applying this method until $i=n$ gives a term $P_{n+1}$ in which everything
  except the $u^ε.0$ has been consumed, so this provides a exhaustive
  pre-trace $ρ$ of $\implem{T}$.

  Now, let $ρ'$ be another exhaustive pre-trace of $\implem{T}$.
  By definition, $ρ$ and $ρ'$ trigger the same actions in $\implem{T}$.
  From this we can deduce that $ρ$ and $ρ'$ contain the same transition
  labels, indeed the actions $\lin{\action}(a)$ are necessarily consumed
  by visible transitions since they are joined together by a composition
  without interaction and the only composition with interaction they are
  involved in is with the names $x_{ab}$ and $y_{ab}$.
  On the other hand, all actions on these names are consumed by internal
  transitions, and for each such name there is exactly one linear input and
  one linear output so there is only one possible internal transition for each
  name.
  As a consequence the sets of actions of $ρ$ and $ρ'$ are the same so $ρ=ρ'$.

  Let us now prove that $ρ^*=T$.
  The only thing we have to check is that the causal order of $ρ$ is the order
  of $T$.
  First consider two events $a,b∈\web{T}$ with $a<b$.
  The action $\lin{\action}(b)$ in $\implem{T}$ is prefixed by $\lin{x}_{ab}$
  (and possibly other actions), and $\linb{x}_{ab}$ is prefixed by
  $\lin{y}_{ab}$, which is itself prefixed by $\lin{\action}(a)$, so the
  transition $\action(a)$ is before the transition $\action(b)$ in $ρ$.
  Then consider two incomparable events $a$ and $b$.
  There is a topological ordering of $\web{T}$ that places $a$ before $b$ and
  another that places $b$ before $a$, so by the construction above we can
  construct an interaction in $ρ$ for each case, which proves that the
  transitions $\action(a)$ and $\action(b)$ are incomparable in the causal
  order of $ρ$.
\end{proof}

This result justifies that $\implem{T}$ is considered as an implementation of
$T$. The proposition below proves that traces are actually the part of
interactions that are observable by interaction.

\begin{proposition}\label{prop-simple-traces}
  For all simple term $P$, $P≃\bigoplus_{ρ∈\Pretraces_e(P)}\implem{ρ^*}$.
\end{proposition}
\begin{proof}
  Let $Q$ be a simple term.
  The term $P|Q$ is simple, so the outcome of a run of this term
  is either $1$ or $0$.
  Let $ρ$ be a run with outcome $1$.
  This implies that no inaction of $P|Q$ is triggered in $ρ$ and
  that each linear action is triggered.

  Call $ρ¹$ the projection of $ρ$ on $P$.
  Formally, $ρ¹$ is obtained from $ρ$ by replacing each transition $1.a$
  with $a$, removing every transition $2.a$ and replacing each transition
  $(1.ι,2.κ)$ with the $α:ι$ that is the left premise of the derivation of
  $(1.ι,2.κ)$.
  Note that the order on $ρ¹$ need not be the restriction to $ρ¹$ of the order
  on $ρ$, it is only a subset of this order.
  Call $ρ²$ the analogous projection on $Q$.

  Since the outcome of $ρ$ is $1$, the pre-traces $ρ¹$ and $ρ²$ are exhaustive
  pre-traces of $P$ and $Q$ respectively.
  Let $ρ'$ be the unique exhaustive pre-trace of $\implem{(ρ¹)^*}$.
  By construction, there is a bijection between the positions of the visible
  actions of $ρ¹$ and those of $ρ'$, which establishes a bijection between
  runs of $P|Q$ with outcome $1$ that project on $P$ as $ρ¹$:
  \[
    \set{θ}{θ∈\Runs(P|Q), \; θ¹=ρ¹, \; s((P|Q)/θ)=1}
  \]
  and runs of $\implem{(ρ¹)^*}|Q$ with outcome $1$:
  \[
    \set{θ}{θ∈\Runs(\implem{(ρ¹)^*}|Q), \; s((\implem{(ρ¹)^*}|Q)/θ)=1}.
  \]
  This bijection preserves outcomes, so we have
  \[
    \sum_{\substack{σ∈\Runs_e(P|Q) \\ σ¹=ρ¹}} s((P|Q)/σ)
    = \outcome{\implem{(ρ¹)^*}|Q}
  \]
  Summing for all potential values of $ρ¹$, i.e. all exhaustive
  pre-traces of $P$, yields
  \[
    \outcome{P|Q}
    = \sum_{ρ¹∈\Pretraces_e(P)}
      \sum_{\substack{σ∈\Runs_e(P|Q) \\ σ¹=ρ¹}} s((P|Q)/σ)
    = \sum_{ρ¹∈\Pretraces_e(P)} \outcome{\implem{(ρ¹)^*}|Q}
  \]
  from which we can conclude.
\end{proof}

We can thus consider traces as terms of the language.
Indeed, given a trace $T$, all simple terms that have a unique exhaustive
pre-trace $ρ$ with $ρ^*=T$ are equivalent to $\implem{T}$.
When precise syntactic information is needed, $T$ used as a term is a
short-hand for $\implem{T}$. 

\begin{theorem}
  Every term is equivalent to a linear combination of traces.
\end{theorem}
\begin{proof}
  By the decomposition of affine actions from
  proposition~\ref{prop-affine-action} we get that every term is equivalent to
  a linear combination of simple terms.
  By proposition~\ref{prop-simple-traces}, each simple term is in turn
  equivalent to a sum of trace implementations.
  The composition of these equivalences, with the module structure of
  $\Procs_\K$, yields a decomposition of every term as a linear combination of
  trace implementations.
\end{proof}


\subsection{Pure trace semantics} 

We can thus define a semantics of processes based on traces, as of
definition~\ref{def-trace}, by reformulating the various constructions for
combinations of traces.
As an example we give a reformulation of testing for traces.
In the definition below, for two traces $T$ and $U$, if $f$ is a function
from $\web{T}$ to $\web{U}$, then $f$ is implicitly extended to a function from
$\web{T}\uplus\Names$ to $\web{U}\uplus\Names$ as the identity over names.

%

\begin{proposition}\label{prop-synchro}
  Let $T$ and $U$ be two traces, $\outcome{T|U}$ is the number of
  synchronizations of $T$ and $U$, that is
  bijections $σ$ from $\web{T}$ to $\web{U}$ such that
  \begin{itemize}
  \item for all $a∈\web{T}$, $p_U(σ(a))=¬p_T(a)$ and $s_U(σ(a))=σ(s_T(a))$,
  \item the relation $\set{(a,b)}{a≤_Tb\text{ or }σ(a)≤_Uσ(b)}$ is acyclic,
  \item for all $(ε,x)∈N_T$, $(¬ε,σ(x))∉N_U$.
  \end{itemize}
\end{proposition}
\begin{proof}
  Let $P=\implem{T}|\implem{U}$.
  By proposition~\ref{prop-simple-traces} we have
  $P≃\bigoplus_{ρ∈\Pretraces_e(P)}\implem{ρ^*}$ hence
  $\outcome{P}=\sum_{ρ∈\Pretraces_e(P)}\outcome{\implem{ρ^*}}$.
  Clearly, for all trace $V$, $\outcome{V}$ is $1$ if $\web{V}=∅$ and
  $0$ otherwise, so $\outcome{P}$ is the number pre-traces $ρ∈\Pretraces_e(P)$
  such that $\web{ρ^*}$ is empty.

  Consider such a pre-trace $ρ$, by definition $ρ$ triggers all linear actions
  in $\implem{T}$ and $\implem{U}$, so $ρ^1$ and $ρ^2$ are the unique
  exhaustive pre-traces of $\implem{T}$ and $\implem{U}$ respectively.
  The relation $\set{(ι,κ)}{(1.ι,2.κ)∈ρ}$ establishes a bijection between
  positions of actions in $\implem{T}$ and $\implem{U}$, which implies a
  bijection $σ:\web{T}→\web{U}$.
  Clearly, for all $a∈\web{T}$, we have $p_U(σ(a))=¬p_T(a)$ since an action
  can only interact with an action of the opposite polarity.
  It is also easy to prove that $s_U(σ(a))=σ(s_T(a))$, since two actions that
  interact are either on the same public name or on private names that are
  unified by the interaction of previous actions.
  Consider two transitions $(ι_a,κ_a)$ and $(ι_b,κ_b)$ in $ρ$, that correspond
  to the pairs of actions $(a,σ(a))$ and $(b,σ(b))$: if $a≤_Tb$
  then the action at $ι_a$ must occur before the action at $ι_b$, so
  $(ι_a,κ_a)≤(ι_b,κ_b)$; the same argument applies if $σ(a)≤_Uσ(b)$, so the
  order $≤_ρ$ contains the orders $≤_T$ and $≤_U$, which proves that the union
  of these orders is acyclic.
  Finally, if there were $ε∈\Pola$ and $x∈\Names\uplus\web{T}$ such that
  $(ε,x)∈N_T$ and $(¬ε,σ(x))∈N_U$, then the run $ρ$ could be extended with an
  interaction between the inactions associated with them, and the outcome would
  be $0$.
  Therefore $σ∈\Sync(T,U)$.

  Reciprocally, let $σ$ be a synchronization of $T$ and $U$.
  The relation $\set{(a,b)}{a≤_Tb\text{ or }σ(a)≤_Uσ(b)}$ is acyclic,
  so there is a non-decreasing enumeration $\web{T}=\{a_1,…,a_n\}$ such that
  $σ(a_1),…,σ(a_n)$ is also non-decreasing.
  Then there is a path $p∈\implem{T}$ that reaches $a_1,…,a_n$ in this order
  and a run $q∈\implem{U}$ that reaches $σ(a_1),…,σ(a_n)$ in this order.
  By combining $p$ and $q$ we get a path $r∈P$.
  Indeed, for each $i$ we have $p(σ(a_i))=¬p(a_i)$ and $s(σ(a_i))=σ(s(a_i))$
  so either $a_i$ and $σ(a_i)$ have the same public name as subject, or their
  subjects are two bound names $z_{s(a_i)}$ and $z_{s(σ(a_i))}$.
  Since $s(a_i)<a_i$ by definition, there is $j<i$ such that $s(a_i)=a_j$ and
  then $s(σ(a_i))=σ(a_j)$, so the subjects of $a_i$ and $σ(a_i)$ are unified
  by the interaction between $a_j$ and $σ(a_j)$.
  In any case, the actions $a_i$ and $σ(a_i)$ can interact.
  The term $P/r$ is the composition with interaction of $\implem{T}/p$ and
  $\implem{U}/q$, and these terms are compositions without interaction of the
  inactions that correspond to $N_T$ and $N_U$ respectively.
  The condition that $(¬ε,σ(x))∉N_U$ for all $(ε,x)∈N_T$ guarantees that no
  further interaction can occur, therefore $r$ is a maximal path of $P$ and
  $s(P/r)=1$.

  These construction establish a bijection between $\Sync(T,U)$ and the runs
  of $\implem{T}$ and $\implem{U}$ with outcome $1$, which proves the expected
  result.
\end{proof}

%
%

We will not develop the trace semantics further here for lack of space, but
the abstract reformulation of outcomes above gives an idea of the
construction: a finite process is interpreted as a linear combination of
traces and all basic operations are defined independently of the semiring
$\K$.
The linear action prefix maps traces to traces,
inactions are basic traces,
composition without interaction is a disjoint union of traces,
composition with interaction maps a pair of traces to a combination of traces
with integer coefficients,
hiding $\new{u}$ maps traces that contain an action on $u$ to $0$, and remove
inactions on $u$ from other traces.

\section{Classic forms of test} 
\label{sec-classic}

By choosing appropriate structures for $\K$, we
can recover the standard may and must testing.
In both cases we have $\K=\implem{0,1,ω}$, where $ω$ represents success.
Table~\ref{table-may-must} show the rules for addition and multiplication for
may and must.
Using this definition it is clear that $P$ and $Q$ are equivalent for may or
must testing if and only if, for all $R$, $\outcome{P|R}=ω$ if and only if
$\outcome{Q|R}=ω$.
Taking for $\K$ the minimal semiring $\{0,1\}$ with $1+1=1$ gives the
framework studied by the author in a previous work~\cite{bef08:apc}, which
also leads to must testing semantics.

\begin{table}
  \centering
  \begin{tabular}{c@{\qquad}c@{\qquad}c}
    may and must &
    may testing &
    must testing \\[.5ex]
    $ \begin{array}{c|ccc}
        ⋅ & 0 & 1 & ω \\ \hline
        0 & 0 & 0 & 0 \\
        1 & 0 & 1 & ω \\
        ω & 0 & ω & ω
      \end{array} $ &
    $ \begin{array}{c|ccc}
        + & 0 & 1 & ω \\ \hline
        0 & 0 & 1 & ω \\
        1 & 1 & 1 & ω \\
        ω & ω & ω & ω
      \end{array} $ &
    $ \begin{array}{c|ccc}
        + & 0 & 1 & ω \\ \hline
        0 & 0 & 1 & ω \\
        1 & 1 & 1 & 1 \\
        ω & ω & 1 & ω
      \end{array} $
  \end{tabular}
  \smallskip
  \caption{Observation semirings for may and must testing.}
  \label{table-may-must}
\end{table}

These semirings share an important property, namely that all elements are
idempotent for addition.
This is an important restriction, in particular it implies that summing
outcomes cannot count the number of successes or failures.
In other words, the “quantitative” part of our testing semantics disappears.
We can remark that this constraint imposes to forget non-interleaving, since
it allows us to decompose everything as totally ordered traces.

\begin{theorem}
  If $∀x,x+x=x$, then $\Procs_\K$ is generated by totally
  ordered traces.
\end{theorem}
\begin{proof}
  We prove the equivalent fact that each trace is equivalent to the sum of all
  its total orderings.
  Let $T$ and $U$ be two traces.
  Call $\?O$ the set of total orders over $\web{T}$ that contain $≤_T$ and for
  each $R∈\?O$, call $T_R$ the trace obtained from $T$ by replacing the order
  with $R$.
  Let $T'=\bigoplus_{R∈\?O}T_R$.

  By proposition~\ref{prop-synchro}, $\outcome{T|U}$ is the number of
  matchings between $T$ and $U$.
  This means that if there are $n$ matchings, then $\outcome{T|U}=1+\cdots+1$
  with $n$ occurrences of $1$.
  By hypothesis $1+1=1$, so $\outcome{T|U}$ is $1$ if there is at least one
  matching and $0$ otherwise.
  By the same argument, for all $R∈\?O$ we have $\outcome{T_R|U}∈\implem{0,1}$,
  hence $\outcome{T'|U}$ is $1$ if there is at least one $R$ such that
  $\outcome{T_R|U}=1$ and $0$ otherwise.

  Assume $\outcome{T|U}=1$, and let $σ$ be a matching between $T$ and $U$.
  Then $σ$ induces an order $≤$ on $\web{T}$ such that $a≤_Tb$ and $σ(a)≤_Uσ(b)$
  both imply $a≤b$.
  Any completion $R$ of $≤$ into a total order yields a total ordering $T_R$
  of $T$ such that $\outcome{T_R|U}=1$, which proves that $\outcome{T'|U}=1$.

  Reciprocally, assume that $\outcome{T'|U}=1$, then there is an $R∈\?O$ such
  that $\outcome{T_R|U}=1$, then there is a matching $σ$ between $T_R$ and
  $U$.
  Since the only difference between $T_R$ and $T$ is the order and $≤_T$ is
  included in $R$, $σ$ is also a matching between $T$ and $U$, hence
  $\outcome{T|U}=1$.
\end{proof}


\bibliographystyle{plain}
\bibliography{eb}

\end{document}